\documentclass[copyright]{eptcs}
\providecommand{\event}{DCFS 2010} 

\usepackage{breakurl}             
\usepackage[all]{xy}
\usepackage{amssymb,amsmath,graphicx,color,extarrows,wrapfig}

\newlength{\btlabelwidth}
\newlength{\btleftmargin}
\newenvironment{btlists}{\begin{list}{{\rm--}}{%
\setlength{\labelwidth}{\btlabelwidth}\setlength{\leftmargin}{\btleftmargin}%
\setlength{\topsep}{0pt plus0.2ex}%
\setlength{\itemsep}{0ex plus0.2ex}%
\setlength{\parsep}{0pt plus0.2ex}}}{\end{list}}
\newcounter{btlistklac}

\newtheorem{theorem}{Theorem}[section]
\newtheorem{corollary}[theorem]{Corollary}

\newenvironment{proof}{{\it Proof. }}{\hspace*{\fill} $\Box$\smallskip }

\def\N{\mathbb{N}}

\def\cL{{\cal L}}
\def\cP{{\cal P}}
\def\Prod{\mathrm{Prod}}
\def\Symb{\mathrm{Symb}}
\def\klPDOL{k\mathrm{{\ell}PD0L}}
\def\klPDTOL{k\mathrm{{\ell}PDT0L}}
\def\klPOL{k\mathrm{{\ell}P0L}}
\def\klPTOL{k\mathrm{{\ell}PT0L}}
\def\klTOL{k\mathrm{{\ell}T0L}}
\def\elPDOL{1\mathrm{{\ell}PD0L}}
\def\elPDTOL{1\mathrm{{\ell}PDT0L}}
\def\elPOL{1\mathrm{{\ell}P0L}}
\def\elPTOL{1\mathrm{{\ell}PT0L}}
\def\lPDOL{\mathrm{{\ell}PD0L}}
\def\lPDTOL{\mathrm{{\ell}PDT0L}}
\def\lPOL{\mathrm{{\ell}P0L}}
\def\lPTOL{\mathrm{{\ell}PT0L}}
\def\lTOL{\mathrm{{\ell}T0L}}
\def\PDOL{\mathrm{PD0L}}
\def\PDTOL{\mathrm{PDT0L}}
\def\POL{\mathrm{P0L}}
\def\PTOL{\mathrm{PT0L}}

\def\Lra{\Longrightarrow}
\def\ra{\rightarrow}
\def\topp#1#2{\genfrac{}{}{0pt}{1}{#1}{#2}}
\def\topps#1#2{\genfrac{}{}{0pt}{2}{#1}{#2}}

\def\Set#1#2{\left\{\: #1\;|\; #2\:\right\}}

\def\Setr#1#2{\left\{\: #1\;\left|\; #2\right.\:\right\}}
\def\Sets#1{\left\{\,#1\,\right\}}

\def\mand{\mbox{ and }}
\def\mfor{\mbox{ for }}

\def\qmand{\quad \mbox{and} \quad}

\DeclareSymbolFont{letters}{OML}{cmm}{m}{it}
\DeclareSymbolFont{symbols}{OMS}{cmsy}{m}{n}

\title{On the Descriptional Complexity of Limited Propagating Lindenmayer Systems}
\author{Bianca Truthe
\institute{Otto-von-Guericke-Universit{\"a}t Magdeburg, Fakult{\"a}t f{\"u}r Informatik\\
PSF 4120, D-39016 Magdeburg, Germany}
\email{truthe@iws.cs.uni-magdeburg.de}
}
\def\titlerunning{On the Descriptional Complexity of Limited Propagating Lindenmayer Systems}
\def\authorrunning{B. Truthe}
\begin{document}
\maketitle

\begin{abstract}
    We investigate the descriptional complexity of limited propagating Lindenmayer 
    systems and their deterministic and tabled variants
    with respect to the number of rules and the number
    of symbols. We determine the decrease of complexity when the generative
    capacity is increased. For incomparable families, we give 
    languages that can be described more efficiently in either of these families
    than in the other.
\end{abstract}

\section{Introduction}

Several generating devices for formal languages have been studied in the literature
with respect to the size of their descriptions (e.\,g., \cite{Gru76}). For sequentially
deriving grammars, the measures number of productions, number of nonterminal symbols,
and number of all symbols have been investigated.

In 1968, Lindenmayer systems (L-systems) have been introduced~(\cite{Lin68}).
In order to model the development of organisms, these devices work in parallel (in one
derivation step, not only one symbol is rewritten as in a sequential grammar but all
symbols are rewritten). For L-systems, the number of tables, the number of active symbols,
and the degree of nondeterminism have been studied as measures of complexity. In \cite{Das04},
the measures number of rules and number of symbols were introduced for L-systems.

Twenty years after the introduction of L-systems, a restricted variant of L-systems
with a partially parallel derivation process has been proposed in \cite{Wae88}. In
these so-called $k$-limited L-systems, only $k$ occurrences of each symbol are replaced
according to some rule. First results on the descriptional complexity of~$k$-limited 
L-systems can be found in \cite{Har09}.

We continue this work and study the relations that were left open in \cite{Har09}
or that have not been optimal yet. In this paper, we confine ourselves to propagating
limited systems.

\section{Definitions}

We assume that the reader is familiar with the basic concepts of formal language theory
(see e.\,g. \cite{handbook}). We recall here some notations used in the paper.

We denote the set of all positive integers by $\N$ and the set of all non-negative integers by $\N_0$.

For an alphabet $V$ (a finite set of symbols), we denote by $V^*$ the set of all words 
over $V$, by $V^+$ the set of all non-empty words over $V$, and by $V^n$ for a natural
number $n\in\N_0$ the set of all words which have the length $n$. We denote the empty 
word by $\lambda$, the length of a word $w$ by $|w|$, and the
number of occurrences of a letter $a$ in a word $w$ by $|w|_a$.
Furthermore, we denote the cardinality of a set $A$ by $|A|$. 

Two sets $X$ and $Y$ are called incomparable, if neither $X\subseteq Y$ nor $Y\subseteq X$ holds.
They are called disjoint if the intersection is empty.

A tabled interactionless Lindenmayer system (L-system for short), abbreviated as $\mathrm{T0L}$ system,
is a triple $G=(V,\cP,\omega)$ where $V$ is an alphabet, $\omega\in V^+$ is called the axiom, and $\cP$ is a 
finite, non-empty set $\Sets{P_1,P_2,\ldots, P_n}$ where $P_i$ (called a table), 
for $1\leq i\leq n$, is a finite subset of $V\times V^*$ such that there is at least one 
element $(a,w)\in P_i$ for each letter $a\in V$. The elements $(a,w)$ in some table are
called productions or rules and are written as $a\ra w$.

A $\mathrm{T0L}$ system $G=(V,\cP,\omega)$ is called an $\mathrm{0L}$ system if $\cP$ contains
only one table. It is called a $\mathrm{DT0L}$ system if every table $\cP$ contains only one rule
for each letter in $V$ and it is called a $\mathrm{D0L}$ system if $\cP$ contains only one table
and the table consists of only one rule for each letter in $V$.

Such an L-system is called propagating, if there is no erasing rule $a\ra\lambda$ in the system
(all rules have the form $a\ra w$ with $a\in V$ and $w\in V^+$).

A word $v\in V^+$ directly derives a word $w\in V^*$ by a system $G$, written as $v\xLongrightarrow[G]{} w$ (we
omit the index if it is clear from the context), if $v=x_1x_2\cdots x_m$ with $m\in\N$, $x_i\in V$
for $1\leq i\leq m$ and $w=y_1y_2\cdots y_m$ with $y_i\in V^*$ for $1\leq i\leq m$ such that
the system $G$ contains a table $P$ which contains all the rules $x_i\ra y_i$ for $1\leq i\leq m$.
Hence, in parallel, every letter of a word is replaced by a word according to the rules of a table.
By $\xLongrightarrow{*}$, we denote the reflexive and transitive closure of $\Lra$. The language generated
by a system $G$ is defined as 
\[L(G)=\Setr{z}{\omega\xLongrightarrow[G]{*} z}.\]

In \cite{Wae88}, a limitation of the parallel rewriting was introduced. For a natural number $k\in\N$,
a $k$-limited $\mathrm{T0L}$ system (shortly written as $\klTOL$ system) is a quadruple $G=(V,\cP,\omega,k)$
where $(V,\cP,\omega)$ is a $\mathrm{T0L}$ system. In a $k$-limited system, exactly $\min\{k,|w|_a\}$
occurrences of any letter $a$ in the word $w$ under consideration are rewritten in a derivation step
(hence, the number of occurrences of a letter that are replaced in each step is limited by $k$).

We only say a $\mathrm{T0L}$ system is limited (shortly written as $\lTOL$ system) if it is 
a $k$-limited system for some number $k\in\N$.

The class of all $k$-limited $\mathrm{T0L}$ systems is written as $\klTOL$. The restricted and 
propagating variants thereof are denoted by $\klPDOL$, $\klPOL$, $\klPDTOL$, $\klPTOL$, and 
without $k$ if the limit is arbitrary. For a class $X$ of L-systems, we write $\cL(X)$ for
the family of languages that is generated by an L-system from~$X$.

As measures of descriptional complexity, we consider the number of rules and the number of symbols.
For an L-system $G$ over an alphabet $V$ with tables $P_1,P_2,\ldots,P_n$ with $n\in\N$ and an
axiom $\omega$, we set
\[\Prod(G)=\sum_{i=1}^n|P_i| \qmand \Symb(G)=|\omega|+\sum_{i=1}^n\sum_{a\ra w\in P_i}(|w|+2).\]
Let $X$ be a class of L-systems. For a language $L\in\cL(X)$, we set
\begin{align*}
\Prod_X(L) &= \min\Set{\Prod(G)}{G\in X \text{ with } L(G)=L} \mand\\
\Symb_X(L) &= \min\Set{\Symb(G)}{G\in X \text{ with } L(G)=L}.
\end{align*}
Hence, the complexity of a language $L$ with respect to a class $X$ of L-systems is the complexity 
of a smallest L-system $G\in X$ that generates the language $L$. If we extend a class $X$ to a 
class $Y$ then the complexity can only become smaller: If $X\subseteq Y$, then $K_X(L)\geq K_Y(L)$
for any language $L\in\cL(X)$ and complexity measure $K\in\Sets{\Prod,\Symb}$.

We now define the complexity relations considered in this paper. Let $X$ and~$Y$ be two classes
of L-systems such that the language families $\cL(X)$ and $\cL(Y)$ are not disjoint and let
$K\in\Sets{\Prod,\Symb}$ be a complexity measure. 

\noindent We write
\begin{itemize} 
\item $X=^K Y$ if $K_X(L)=K_Y(L)$ holds for any language $L\in\cL(X)\cap\cL(Y)$ (the complexities
      are equal),
\item $X>^K Y$ if there is a sequence of languages $L_m\in\cL(X)\cap\cL(Y)$ for $m\in\N$, such
      that\linebreak
      $K_X(L_m)-K_Y(L_m)\geq c\cdot m$ for a constant $c\in\N$ (the difference of the 
      complexities can be arbitrarily large),
\item $X\gg^K Y$ if there is a sequence of languages $L_m\in\cL(X)\cap\cL(Y)$ for \hbox{$m\in\N$}, such
      that\linebreak
      $\lim\limits_{m\to\infty}\frac{K_Y(L_m)}{K_X(L_m)}=0$ (asymptotically, the complexity 
      using $X$ grows faster than using $Y$),
\item $X\ggg^K Y$ if there is a sequence of languages $L_m\in\cL(X)\cap\cL(Y)$, $m\in\N$, and
      a constant $c\in\N$ such that $K_Y(L_m)\leq c$ and $K_X(L_m)\geq m$.
\end{itemize}
From these definitions, we obtain that $X\ggg^K Y$ implies $X\gg^K Y$ and that also $X\gg^K Y$ implies $X>^K Y$
for $K\in\Sets{\Prod,\Symb}$.

For each natural number $c$, there are only finitely many L-systems $G$ (upto renaming the 
symbols) for which $\Symb(G)\leq c$ holds. Hence, there is no class $X$ of L-systems that 
generates infinitely many languages $L_n$ with $\Symb_X(L_n)\leq c$. Thus, there exist no two
classes $X$ and $Y$ with the relation $Y\ggg^\Symb X$.

In all cases throughout this paper, we obtain the relation $X\gg^{\Symb} Y$ whenever we also obtain
$X\ggg^{\Prod} Y$. Then, we also shortly write $X\ggg Y$.
Further, if two classes $X$ and $Y$ are in the same relation $\rhd$ with respect to both measures
$\Prod$ and $\Symb$, hence, $X \rhd^{\Prod} Y$ and $X \rhd^{\Symb} Y$ for a symbol
$\rhd\in\{\gg,>,=\}$, then we write $X \rhd Y$.

\section{On 1-limited systems}

Regarding 1-limited propagating L-systems, the following hierachy is known (\cite{Har09}).

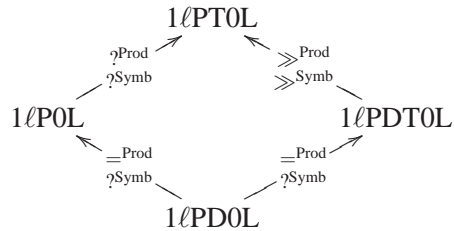
\begin{figure}[h]
\centerline{$\scalebox{1}{\xymatrix{
& {\elPTOL} &\\
{\elPOL}\ar[ru]|-{\topp{?^{\Prod}}{?^{\Symb}}} & & {\elPDTOL}\ar[lu]|-{\topp{\gg^{\Prod}}{\gg^{\Symb}}}\\
& {\elPDOL}\ar[lu]|-{\topp{=^{\Prod}}{?^{\Symb}}}\ar[ru]|-{\topp{=^{\Prod}}{?^{\Symb}}} &
}}$}
\caption{Relations for 1-limited systems}\label{fig-1-old}
\end{figure}

An arrow from a class $X$ to a class $Y$ with a label $R$ is to be read as the relation $X R Y$. If the label 
contains a question mark, then the relation was not given in \cite{Har09}.
In this section, we prove relations for all these cases and also relations between
the classes $\elPOL$ and $\elPDTOL$.

\begin{theorem}\label{th-1PD0L-1PDT0L}
The relation $\elPDOL=^\Symb \elPDTOL$ holds.
\end{theorem}
\begin{proof}
Let $G=(V,\{P\},\omega,1)$ be a $\elPDOL$ system which is minimal with respect to the number of symbols
for the language $L=L(G)$ with $V=\{a_1,a_2,\dots, a_n\}$ and $P=\Set{a_i\to w_{a_i} }{ 1\leq i\leq n}$.
Further, let $H=(V,\{P_1,P_2,\ldots P_m\},\omega_H,1)$ be a minimal $\elPDTOL$ system for the
language $L$.

Let $\Omega_n$ be the set of all words that are derived by $G$ in $n$ steps from the axiom:
\[\Omega_0 = \{\omega\},\
\Omega_n = \Set{ w }{ \omega\Lra^n_G w } = \Set{ w }{ \mbox{there is $u\in\Omega_{n-1}$ with $u\Lra_G w$}}.
\]

For all $n\geq 0$, we have:
\begin{btlists}
\item The set $\Omega_n$ is not empty.
\item All words in $\Omega_n$ contain the same number of letters for each letter of~$V$
(during the derivation, the same rules are applied -- only at different positions). 
As a consequence, all words in $\Omega_n$ have the same length. Let it be denoted by $l_n$. 
The set of all occurring letters is denoted by $\alpha_n$.
\end{btlists}

Since $G$ is propagating, we have $l_0\leq l_1\leq l_2\leq\cdots\leq l_i\leq\cdots$

For each word $w_1$, from which a word $w_2$ is derived by $H$ in one step, there are
words~$w_3$ and $w_4$ such that the following holds:
\begin{btlists}
\item $w_1$ and $w_3$ belong to the same set $\Omega_p$ for a number $p\geq 0$,
\item $w_3\Lra_G w_4$ and
\item $|w_4|\leq |w_2|$.
\end{btlists}

This implies
\[|w_4|=|w_3|+\sum\limits_{\topps{i=1}{|w_3|_{a_i}>0}}^n(|w_{a_i}|-1)\]
(each letter $a_i$ appearing in $w_3$ is replaced once by the corresponding 
word $w_{a_i}$; hence, $|w_{a_i}|-1$ letters are added).
Since $|w_2|\geq |w_4|$, we also have 
\[|w_2|\geq |w_3|+\sum_{\topps{i=1}{|w_3|_{a_i}>0}}^n(|w_{a_i}|-1).\]

Since $w_3$ and $w_1$ belong to the same set $\Omega_p$, we have $|w_3|=|w_1|=l_p$ and
the words $w_3$ and $w_1$ consist of the same letters (the set of the appearing letters
is~$\alpha_p$). 

Hence,
\[|w_2| \geq |w_1|+s \mbox{ with }
  s = \sum_{\topps{i=1}{a_i\in\alpha_p}}^n(|w_{a_i}|-1).
\]

Let $P_j$ be that table by which $w_2$ is derived from $w_1$ in $H$. Then we have, for the
number $|P_j|$ of the symbols occurring in $P_j$,
\[|P_j|\geq s+\sum_{i=1}^n 3\]
(for each letter $a_i\in V$, there is a rule with at least three symbols; furthermore, the $s$ new
letters (the difference between $w_2$ and $w_1$) have to be generated and each rule is used at 
most once).
In other words, we have
\[|P_j|\geq \sum_{\topps{i=1}{a_i\in\alpha_k}}^n(|w_{a_i}|+2)+\sum_{\topps{i=1}{a_i\notin\alpha_k}}^n 3.\]
If $P_j$ is applied to words from $\Omega_q$ and $\Omega_r$ for any $q\geq 0$ and $r\geq 0$, then this inequality holds 
for $p=q$ as well as for $p=r$. Thus,
\[|P_j|\geq \sum_{\topps{i=1}{a_i\in\alpha_q\cup\alpha_r}}^n(|w_{a_i}|+2)+\sum_{\topps{i=1}{a_i\notin\alpha_q\cup\alpha_r}}^n 3.\]
Let $A$ be the union of the sets $\alpha_p$ for those $p\geq 0$ for which $P_j$ is applied to a word
from $\Omega_p$. Then
\[|P_j|\geq \sum_{\topps{i=1}{a_i\in A}}^n(|w_{a_i}|+2)+\sum_{\topps{i=1}{a_i\notin A}}^n 3.\]
Since $P_j$ is applied to every word, we obtain
$A=\bigcup\limits_{p\geq 0}\alpha_p =V$.
Thus,
\[|P_j|\geq \sum_{i=1}^n(|w_{a_i}|+2)=\Symb(G)-|\omega|.\]
Together, this yields
\[\Symb(H)=|\omega_H|+\sum_{j=1}^m |P_j|\geq |\omega_H|+|P_1|\geq|\omega_H|+\Symb(G)-|\omega|=\Symb(G).\]

Since each $\elPDOL$ system is also a $\elPDTOL$ system, we have $\Symb(H)\leq \Symb(G)$ on the other hand. This yields the claim.
\end{proof}

The proof of the previous theorem can be changed such that $H$ is a $\elPOL$ system (and $P_j$ is the table of the system). 
Then we obtain the following result.

\begin{corollary}\label{cor-1PD0L-1P0L}
The relation $\elPDOL=^\Symb \elPOL$ holds.
\end{corollary}

Next, we prove relations between $\elPOL$ and $\elPTOL$ systems.
\begin{theorem}\label{th-1P0L-1PT0L}
The relation $\elPOL\ggg \elPTOL$ holds.
\end{theorem}
\begin{proof}
Let $m\in\N$, $V=\{a,b,c,d,e\}$, and 
\[L_m=\{e\}\cup\Set{a^nx_1x_2\cdots x_md^n}{n\geq 1,\ x_i\in\{b,c\},\ 1\leq i\leq m}.\]
The $\elPTOL$ system $G_m=(V,\{P_1,P_2\},e,1)$ with
\begin{align*}
P_1 &= \Sets{a\ra a, b\ra c, c\ra c, d\ra d, e\ra ab^md} \mand\\
P_2 &= \Sets{a\ra aa, b\ra b, c\ra c, d\ra dd, e\ra e}
\end{align*}
generates the language $L_m$ (the first table generates all words $awd$ for $w\in\{b,c\}^*$ with $|w|=m$; the second
table increases the number of occurrences of $a$ and $d$). Since the complexities are
\hbox{$\Prod(G_m)=10$} and $\Symb(G_m)=m+34$, we obtain
$\Prod_{\elPTOL}(L_m)\leq 10$ and $\Symb_{\elPTOL}(L_m)\leq m+34$.
Each language $L_m$ is also generated by a $\elPOL$ system, for instance,
by $G'_m=(V,\{P\},e,1)$ with
\begin{align*}
P &= \Set{e\ra awd}{w\in\{b,c\}^*,|w|=m}\cup\Sets{a\ra aa, b\ra b, c\ra c, d\ra dd}
\end{align*}
(the first application of a rule yields all words $awd$ for $w\in\{b,c\}^*$ with $|w|=m$;
from the second application on, the number of occurrences of $a$ and $d$ is increased).

Now let $H_m=(V,\{P_m\},\omega_m,1)$ be a minimal $\elPOL$ system for $L_m$. Since~$H_m$
is propagating, $\omega_m$ is the shortest word of $L_m$: $\omega_m=e$. This word has to derive 
another word of $L_m$ (otherwise only $e$ is generated). Hence, $P_m$ contains a rule
$e\ra ax_1x_2\cdots x_md$ with $x_i\in\{b,c\}$ for $1\leq i\leq m$ (if $e$ derives only
longer words, the words of length $m+2$ are not generated). Since the number of occurrences of~$a$ in the
beginning of a word in $L_m$ is unbounded, there must be a rule that increases the number. 
This cannot be done by $b$, $c$, or $d$, because then an $a$ would appear at a wrong
position. Hence, it can only be done by $a$. If the rule for $a$ contains other letters 
than $a$, then we obtain words that are not in $L_m$. Thus,~$P_m$ contains a rule $a\ra a^i$
for an integer $i\geq 2$. If two different rules $a\ra a^i$ and $a\ra a^j$ exist, then two
different words $a^iw'$ and $a^jw'$ could be generated but they are not both in $L_m$.
Hence, there is only one rule $a\ra a^i$ in $P_m$. The same argumentation holds for the rules
of $d$. Hence, the only rule for $d$ is~$d\ra d^i$ (if the rule would be $d\ra d^j$ for
a $j$ different from $i$, then the word $ax_1x_2\cdots x_md$ would derive a 
word $a^{i-1}x'_1x'_2\cdots x'_md^{j-1}\notin L_m$). The only possible rules for $b$ and $c$ 
are $b\ra b$ or $b\ra c$ and $c\ra c$ or~$c\ra b$, otherwise a word would be generated that
does not belong to $L_m$. 

Let $w=ax_1x_2\cdots x_md\in L_m$ be a word that is derived from $e$ in one step. Then all
words derived in one or more steps from $w$ contain more than one $a$ (because the only
rule for $a$ is $a\ra a^i$ with $i\geq 2$). Hence, all words with only one $a$ have to
be derived directly from $e$. Hence, $P_m$ contains at least all rules 
$e\ra ax_1x_2\cdots x_md$ with $x_i\in\{b,c\}$ for $1\leq i\leq m$. These are $2^m$ rules
with $m+4$ symbols each.

Hence, $\Prod_{\elPOL}(L_m)\geq 2^m$ and $\Symb_{\elPOL}(L_m)\geq 2^m(m+4)$
which yields $\elPOL\ggg^\Prod \elPTOL$ and $\elPOL\gg^\Symb \elPTOL$.
\end{proof}

The two classes of $\elPOL$ systems and $\elPDTOL$ systems are incomparable. 
However, the language classes are not disjoint. Hence, we can also search for
relations between incomparable classes.

\begin{theorem}\label{th-1P0L-1PDT0L}
The relations $\elPDTOL\gg^K \elPOL$ for a complexity measure $K\in\{\Prod,\Symb\}$ as well as 
$\elPOL\ggg^\Prod \elPDTOL$ and $\elPOL\gg^\Symb \elPDTOL$
are valid.
\end{theorem}
\begin{proof}
The first statement was shown in \cite{Har09}, although not explicitly mentioned 
(the $\elPTOL$ system used in the proof of 
the relation $\elPDTOL\gg^K \elPTOL$ for $K\in\{Prod,\Symb\}$ is also
a $\elPOL$ system).

The other two results follow from the proof of Theorem~\ref{th-1P0L-1PT0L} because
the $\elPTOL$ system used is also a $\elPDTOL$ system.
\end{proof}

The results for 1-limited propagating L-systems can be seen in the following figure.

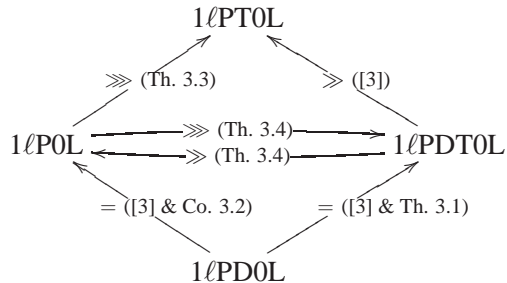
\begin{figure}[h]
\centerline{$\scalebox{1}{\xymatrix@=12mm{
& {\elPTOL} &\\
{\elPOL}\ar[ru]|-{\text{~ ~ ~ ~}\ggg~\text{\rm (Th. \ref{th-1P0L-1PT0L})}}
  \ar@<2pt>@/^3pt/[rr]|-{\ggg~\text{\rm (Th. \ref{th-1P0L-1PDT0L})}} & & 
  {\elPDTOL}\ar[lu]|-{\text{~ ~ ~}\gg~\text{\rm (\cite{Har09})}}
  \ar@<2pt>@/^3pt/[ll]|-{\gg~\text{\rm (Th. \ref{th-1P0L-1PDT0L})}}\\
& {\elPDOL}\ar[lu]|-{\text{~ ~ ~ ~ ~ ~ ~}=~\text{\rm (\cite{Har09} \& Co. \ref{cor-1PD0L-1P0L})}}
\ar[ru]|-{\text{~ ~ ~ ~ ~ ~ ~ ~ ~ ~}=~\text{\rm (\cite{Har09} \& Th. \ref{th-1PD0L-1PDT0L})}} &
}}$}
\caption{Results for 1-limited systems}\label{fig-1-new}
\end{figure}

In brackets behind a relation, you find a link to the corresponding proof.

If a sequence of languages $L_m$ is generated by $\elPOL$ systems or $\elPTOL$ systems
with a constant number of rules, then the languages $L_m$ can also be generated 
by $\elPDTOL$ systems with a constant number of rules. As a consequence, all relations
mentioned above are tight.

\section{On higher limited systems}

Let $k\geq 2$. Regarding $k$-limited propagating L-systems, the following hierachy is known (\cite{Har09}).

\begin{wrapfigure}{r}{12.6cm}
\centerline{$\scalebox{1}{\xymatrix{
& {\klPTOL} &\\
{\klPOL}\ar[ru]|-{\topp{\ggg^{\Prod}}{\gg^{\Symb}}} & & {\klPDTOL}\ar[lu]|-{\topp{{\ggg}^{\Prod}}{{\gg}^{\Symb}}}\\
& {\klPDOL}\ar[lu]|-{\topp{=^{\Prod}}{?^{\Symb}}}\ar[ru]|-{\topp{=^{\Prod}}{?^{\Symb}}} &
}}$}
\caption{Relations for $k$-limited systems}\label{fig-k-old}
\end{wrapfigure}
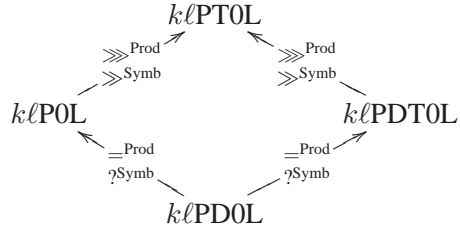

An arrow from a class $X$ to a class $Y$ with a label $R$ is to be read as the 
relation $X R Y$. If the label contains a question mark, then the relation was 
not given in~\cite{Har09}.
In this section, we give relations for these cases and also relations between
the classes $\klPOL$ and $\klPDTOL$.

\begin{theorem}\label{th-kPD0L-kPDT0L_kP0L}
We have $\klPDOL=^\Symb \klPDTOL$ and $\klPDOL=^\Symb \klPOL$ for every $k\geq 2$.
\end{theorem}
\begin{proof}
The proof is in both cases similar to the proof of Theorem~\ref{th-1PD0L-1PDT0L}.
\end{proof}

The two classes of $\klPOL$ systems and $\klPDTOL$ systems are incomparable. However, the classes of
the generated languages are not disjoint. There are languages in the intersection 
that can be described more efficently by $\klPDTOL$ systems than by $\klPOL$ systems.

\begin{theorem}\label{th-kP0L-kPDT0L}
The relation $\klPOL\ggg \klPDTOL$ is valid for $k\geq 2$.
\end{theorem}
\begin{proof}
We generalize the proof of Theorem~\ref{th-1P0L-1PT0L}.
Let $k\geq 2$, $m\in\N$, and
\[L_m\!=\!\{e\}\cup\Set{a^{kn}wd^{kn}}{n\geq 1, w\in\{b,c\}^{km}, |w|_b=kj \mfor 0\leq j\leq m}.\]
This language is generated by the $\klPDTOL$ system 
$G_m=(V,\{P_1,P_2\},e,k)$ with
\hbox{$V=\{a,b,c,d,e\}$} and
$P_1 = \Sets{a\ra a, b\ra c, c\ra c, d\ra d, e\ra a^kb^{km}d^k}$,
$P_2 = \Sets{a\ra aa, b\ra b, c\ra c, d\ra dd, e\ra e}$.
From this system, we obtain the relations
$\Prod_{\klPDTOL}(L_m)\leq 10$ and $\Symb_{\klPDTOL}(L_m)\leq (m+2)k+32$.

Each language $L_m$ is also generated by a $\klPOL$ system, for instance,
by $G'_m=(V,\{P\},e,k)$ with
\[P = \Setr{e\ra a^kwd^k}{w\in\{b,c\}^{km}, |w|_b=kj \mfor 0\leq j\leq m}
  \cup\Sets{a\ra aa, b\ra b, c\ra c, d\ra dd}.\]

By a similar argumentation as in the proof of Theorem~\ref{th-1P0L-1PT0L}, the rules 
for $e$ being adopted to $k$, we obtain that a minimal $\klPOL$ system contains at 
least all rules $e\ra a^kwd^k$ where $w\in\{b,c\}^{km}$ and
the number of $b$s in $w$ is a multiple of $k$.
Thus, 
\[\Prod_{\klPOL}(L_m)\geq 2^m \qmand \Symb_{\klPOL}(L_m)\geq 2^m((m+2)k+2)\]
which gives the relations $\klPOL\ggg^\Prod \klPDTOL$ and $\klPOL\gg^\Symb \klPDTOL$. 
\end{proof}

The converse also holds. In the intersection $\cL(\klPOL)\cap\cL(\klPDTOL)$, there are languages  
that can be described more efficently by $\klPOL$ systems than by $\klPDTOL$ systems.

\begin{theorem}\label{th-kPDT0L-kP0L}
The relation $\klPDTOL\ggg \klPOL$ is valid for $k\geq 2$.
\end{theorem}
\begin{proof}
Let $k\geq 2$. Further, let $m\in\N$ be a natural number, $V=\{a,b,c\}$, and
\[L_m=\{c\}\cup\Set{x_1x_2\cdots x_{km}}{x_i\in\{a,bb\},\ 1\leq i\leq km}.\]
The $\klPOL$ system $G_m=(V,\{P\},c,k)$ with
$P = \Sets{a\ra a, a\ra bb, b\ra b, c\ra a^{km}}$
generates the language $L_m$ (in each step, an arbitrary number $j$ of $a$s in a word
with $0\leq j\leq k$ can be choosen to be changed to $bb$).
As $\Prod(G_m)=4$ and $\Symb(G_m)=km+13$, we obtain
$\Prod_{\klPOL}(L_m)\leq 4$ and $\Symb_{\klPOL}(L_m)\leq km+13$.

Since $L_m\setminus\{c\}$ is finite, there is also a $\klPDTOL$ system that generates
the language (all words are derived from $c$).

Let $H_m$ be a minimal $\klPDTOL$ system. The axiom is $c$ because it is the shortest word of $L_m$.
For each word $z\in L_m\setminus\{c\}$, the equation $|z|_a+\frac{1}{2}|z|_b=km$ holds.
Hence, the only possible rule for $b$ in any table is $b\ra b$; the only rules for
$a$ are $a\ra a$ and $a\ra bb$. The words $a^ibba^{km-1-i}$ ($0\leq i\leq km-1$) 
cannot be derived from other words of $L_m\setminus\{c\}$. Thus, $H_m$ contains at 
least those $km$ rules $c\ra a^ibba^{km-1-i}$ and the~$km\cdot (km+3)$ symbols involved. 

Hence, $\Prod_{\klPDTOL}(L_m)\geq km$ and $\Symb_{\klPDTOL}(L_m)\geq km(km+3)$.
This leads to the relations $\klPDTOL\ggg^\Prod \klPOL$ and $\klPDTOL\gg^\Symb \klPOL$.
\end{proof}

The results for $k$-li\-mi\-ted propagating L-systems can be seen in the following figure.

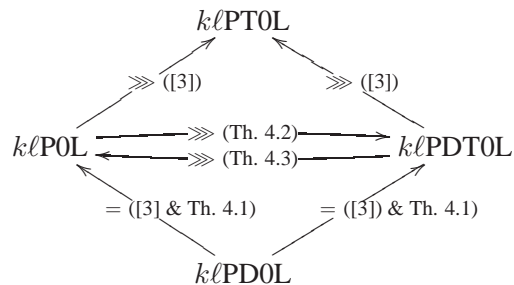
\begin{figure}[h]
\centerline{$\scalebox{1}{\xymatrix@=12mm{
& {\klPTOL} &\\
{\klPOL}\ar[ru]|-{\text{~ ~ ~ ~}\ggg~\text{\rm (\cite{Har09})}}
  \ar@<2pt>@/^3pt/[rr]|-{\ggg~\text{\rm (Th. \ref{th-kP0L-kPDT0L})}} 
  & & 
  {\klPDTOL}\ar[lu]|-{\text{~ ~ ~}\ggg~\text{\rm (\cite{Har09})}}
  \ar@<2pt>@/^3pt/[ll]|-{\ggg~\text{\rm (Th. \ref{th-kPDT0L-kP0L})}}
  \\
& {\klPDOL}\ar[lu]|-{\text{~ ~ ~ ~ ~ ~ ~}=~\text{\rm (\cite{Har09} \& Th. \ref{th-kPD0L-kPDT0L_kP0L})}}
\ar[ru]|-{\text{~ ~ ~ ~ ~ ~ ~ ~ ~ ~}=~\text{\rm (\cite{Har09}) \& Th. \ref{th-kPD0L-kPDT0L_kP0L})}} &
}}$}
\caption{Results for $k$-limited systems}\label{fig-k-new}
\end{figure}
In brackets behind a relation, you find a link to the corresponding proof.
All relations mentioned above are tight.

\section{On arbitrarily limited systems}

Regarding limited propagating L-systems, the following hierachy is known (\cite{Har09}).

\begin{wrapfigure}{r}{10.5cm}
\centerline{$\scalebox{1}{\xymatrix{
& {\lPTOL} &\\
{\lPOL}\ar[ru]|-{\topp{\ggg^{\Prod}}{\gg^{\Symb}}} & & {\lPDTOL}\ar[lu]|-{\topp{\ggg^{\Prod}}{\gg^{\Symb}}}\\
& {\lPDOL}\ar[lu]|-{\topp{=^{\Prod}}{?^{\Symb}}}\ar[ru]|-{\topp{=^{\Prod}}{?^{\Symb}}} &
}}$}
\caption{Relations for limited systems}\label{fig-l-old}
\end{wrapfigure}
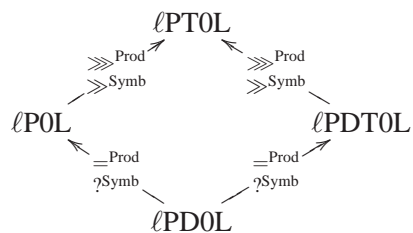

In this section, we prove relations for the open cases and also relations between
the classes $\lPOL$ and $\lPDTOL$.

The proofs for $k$-limited systems cannot directly be used because, for some $k$-limited
system of one kind~$X$ of L-systems, there can be a minimal $m$-limited system of 
another kind $Y$ with $m\not=k$ which has other properties than a minimal $k$-limited 
system of kind $Y$.

\begin{theorem}\label{th-lPD0L-lPDT0L_lP0L}
The relations $\lPDOL\gg^\Symb X$ hold for $X\in\{\lPDTOL,\lPOL\}$.
\end{theorem}
\begin{proof}
Let $m\in\N$, $V=\{a,b,c,d\}$, and
\[L_m=\Set{a^{m+im^2}b}{i\geq 0}\cup\Set{a^{m+im^2+1}c}{i\geq 0}\cup\Sets{d}.\]
Consider an $\lPDOL$ system $G_m=(V,\{P_{G_m}\},\omega,k)$ generating the language~$L_m$.
The only derivation in~$G_m$ is
\begin{multline*}
d\Lra a^mb\Lra a^{m+1}c\Lra a^{m+m^2}b\Lra a^{m+m^2+1}c\Lra \cdots\\
\cdots\Lra a^{m+im^2}b\Lra a^{m+im^2+1}c\Lra a^{m+(i+1)m^2}b\Lra a^{m+(i+1)m^2+1}c\Lra\cdots
\end{multline*}
because $G_m$ is propagating. From this derivation, we obtain that $G_m$ is minimal if
$k=1$ and the rule set $P_{G_m}$ is $\Sets{a\ra aa, b\ra c, c\ra a^{m^2-2}b, d\ra a^mb}$.
Hence, the symbol complexity of $L_m$ is $\Symb_{\lPDOL}(L_m)=m^2+m+12$.

The language $L_m$ can also be generated by an $m$ limited $\PDTOL$ system
\[H_m=(V,\{P_{m,1},P_{m,2}\},d,m)\]
with $P_{m,1} = \Sets{a\ra aa^m, b\ra b, c\ra c, d\ra a^mb}$
and $P_{m,2} = \Sets{a\ra a, b\ra ac, c\ra c, d\ra d}$.
The first table derives from $d$ the word $a^mb$ and from a word $a^px\in L_m$ with $x\in\{b,c\}$ the
word $a^{p+m^2}x$ (hence, every second word is generated). The second table derives from
a word~$a^{m+im^2}b$ with $i\geq 0$ the word~$a^{m+im^2+1}c$ and leaves the other words unchanged.

Hence, we obtain for the symbol complexity of $L_m$ with respect to $\lPDTOL$ systems 
\[\Symb_{\lPDTOL}(L_m)\leq 2m+26\]
which yields the relation $\lPDOL\gg^\Symb\lPDTOL$.

The language $L_m$ can also be generated by an $m$ limited $\POL$ system
$I_m=(V,\{P_{I_m}\},d,m)$ with
$P_{I_m} = \Sets{a\ra aa^m, b\ra b, c\ra c, d\ra a^mb, d\ra a^{m+1}c}$.
In this system, we obtain the words $a^mb$ and $a^{m+1}c$ from $d$ and then, 
by the other rules, from each word every second word.

Hence, we obtain $\Symb_{\lPOL}(L_m)\leq 3m+17$ for the symbol complexity of $L_m$ with 
respect to $\lPOL$ systems which yields the relation $\lPDOL\gg^\Symb\lPOL$.
\end{proof}

The two classes of $\lPOL$ systems and $\lPDTOL$ systems are incomparable. However, 
there are languages in the intersection of the classes that can be described more
efficently by $\lPDTOL$ systems than by $\lPOL$ systems and vice versa.

\begin{theorem}\label{th-lP0L-lPDT0L-lP0L}
The relations $\lPOL\ggg \lPDTOL$ and $\lPDTOL\ggg \lPOL$ hold.
\end{theorem}
\begin{proof}
Let $m\in\N$, $V=\{a,b,c,d,e\}$, and
\[L_m=\{e\}\cup\Set{a^nx_1x_2\cdots x_md^n}{n\geq 1,\ x_i\in\{b,c\},\ 1\leq i\leq m}.\]
Every language $L_m$ is generatable by a $\elPDTOL$ system $G_m$ -- as shown in the proof of 
Theorem~\ref{th-1P0L-1PT0L}. From this proof, we obtain 
$\Prod_{\lPDTOL}(L_m)\leq 10$ and $\Symb_{\lPDTOL}(L_m)\leq m+34$.
Further, we know from that proof that each language~$L_m$ can also be generated 
by a $\elPOL$ system~$H_m$. The argumentation on the minimal system does not depend 
on the limit $k$. Hence, any limited $\POL$ system has at least $2^m$ rules 
and $2^m(m+4)$ symbols.

Thus, $\lPOL\ggg^\Prod \lPDTOL$ and $\lPOL\gg^\Symb \lPDTOL$.

To prove the other relation, let $m\in\N$, $V=\{a,b,c,d\}$, and
\[
L_m=\Sets{d}\cup\{\: w \;|\; w=x_1x_2\cdots x_{2m},\ x_i\in\{a,bb,ccc\}, 1\leq i\leq 2m,\\
|w|_a=2n,\ 0\leq n\leq m \:\}.
\]

The language $L_m$ can be generated by a $2\lPOL$ system $G_m=(V,\{P\},d,2)$ with
the rule set being \hbox{$P=\Sets{a\ra bb, a\ra ccc, b\ra b, c\ra c, d\ra a^{2m}}$}.
Hence, we obtain for the complexities 
\[\Prod_{\lPOL}(L_m)\leq 5 \qmand \Symb_{\lPOL}(L_m)\leq 2m+18.\]

The language $L_m\setminus\{d\}$ is finite. Thus, $L_m$ can be generated by
a limited $\PDTOL$ system with a table for each rule $d\ra w$ where $w\in L_m\setminus\{d\}$.

For each word $z\in L_m\setminus\{d\}$, the equation $|z|_a+\frac{1}{2}|z|_b+\frac{1}{3}|z|_c=2m$ holds.
Hence, the only possible rules in a minimal $\lPDTOL$ system $H_m$ 
are $a\ra a$, $a\ra bb$, $a\ra ccc$, $b\ra b$, and $c\ra c$.

If $H_m$ is a $\elPDTOL$ system, then the word $a^{2m}$ can derive the word $bba^{2m-1}$
or~$ccca^{2m-1}$ which are not in $L_m$ unless the only rule for $a$ is $a\ra a$. But 
then every word $w\in L_m\setminus\{d\}$
has to be derived from~$d$. This yields more than $m$ rules and more than $m^2$ symbols.

If $H_m$ is a $\klPDTOL$ system for some $k\geq 2$, then any word $x_1x_2\cdots x_{2m}$
where one subword $x_i$ is~$bb$, one subword $x_j$ is $ccc$, and all 
other subwords are~$a$ can only be derived from $d$ (if $bb$ or $ccc$ is derived from $a$,
then a second subword~$bb$ or~$ccc$ must exist, because there is only one rule for $a$
and $k\geq 2$). This also yields more than $m$ rules and more than $m^2$ symbols.

Hence, we obtain $\Prod_{\lPDTOL}(L_m)\geq m$ and $\Symb_{\lPDTOL}(L_m)\geq m^2$.
This gives us the relations $\lPDTOL\ggg^\Prod \lPOL$ and $\lPDTOL\gg^\Symb \lPOL$.
\end{proof}

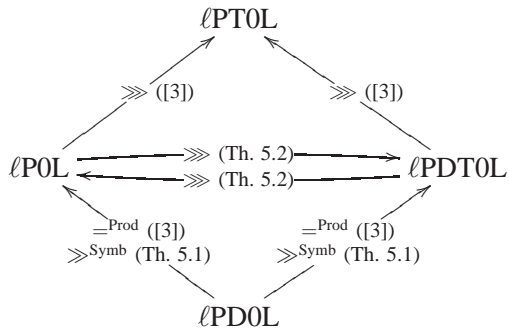
\begin{wrapfigure}{r}{10.1cm}
\centerline{$\scalebox{1}{\xymatrix@=15mm{
& {\lPTOL} &\\
{\lPOL}\ar[ru]|-{\text{~ ~ ~ ~}\ggg~\text{\rm (\cite{Har09})}}
  \ar@<2pt>@/^3pt/[rr]|-{\ggg~\text{\rm (Th. \ref{th-lP0L-lPDT0L-lP0L})}}
  & & 
  {\lPDTOL}\ar[lu]|-{\text{~ ~ ~ ~}\ggg~\text{\rm (\cite{Har09})}}
  \ar@<2pt>@/^3pt/[ll]|-{\ggg~\text{\rm (Th. \ref{th-lP0L-lPDT0L-lP0L})}}
  \\
& 
{\lPDOL}\ar[lu]|-{\topp{=^{\Prod}~\text{\rm (\cite{Har09})}}{\gg^{\Symb}~\text{\rm (Th. \ref{th-lPD0L-lPDT0L_lP0L})}}}
\ar[ru]|-{\topp{=^{\Prod}~\text{\rm (\cite{Har09})}}{\gg^{\Symb}~\text{\rm (Th. \ref{th-lPD0L-lPDT0L_lP0L})}}} &
}}$}
\caption{Results for limited systems}\label{fig-l-new}
\end{wrapfigure}
The results for limited propagating L-systems can be seen in Figure~\ref{fig-l-new}.
In brackets behind a relation, you find a link to the corresponding proof.
All relations mentioned above are tight.

In the remaining part of this section, we investigate the relations 
between $k$-limited and arbitrarily limited propagating
L-systems. 

The relations that are already known (cf.~\cite{Har09}) are to be seen in Figure~\ref{fig-kl-old}. These 
relations hold for $k\geq 1$.

\begin{figure}[ht]
\centerline{$\xymatrix@C=4mm@R=9mm{
& & & & & {\lPTOL} & \\
& {\klPTOL}\ar@{-->}[rrrru]|-{\topp{>^{\Prod}}{>^{\Symb}}} 
  & & & {\lPOL}\ar[ru] & & {\lPDTOL}\ar[lu]\\
{\klPOL}\ar[ru]\ar@{-->}[rrrru]|-(.8){\topp{>^{\Prod}}{>^{\Symb}}} & 
  & {\klPDTOL}\ar[lu]\ar@{-->}[rrrru]|-(.25){\topp{?^{\Prod}}{?^{\Symb}}} 
  & & & {\lPDOL}\ar[lu]\ar[ru] &\\
& {\klPDOL}\ar@{-->}[rrrru]|-{\topp{=^{\Prod}}{?^{\Symb}}}\ar[lu]\ar[ru] & & & & &
}$}
\caption{Relations between $k$-limited and limited propagating L-systems}\label{fig-kl-old}
\end{figure}
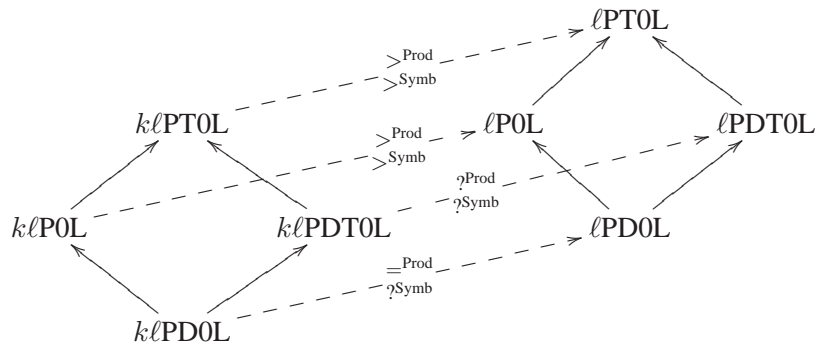

For $k=1$, the relations $\klPDTOL\gg^K\lPDTOL$ for $K\in\Sets{\Prod,\Symb}$ are
given in \cite{Har09}. We now prove relations for the remaining cases and give stronger results
for the existing ones.

\begin{theorem}\label{th-kPDOL-lPDOL}
The relation $\klPDOL\gg^\Symb \lPDOL$ holds for $k\geq 1$.
\end{theorem}
\begin{proof}
Let $k\geq 1$. For $m\in\N$, let $L_m=\{ a^{5mk(1+nm)} \mid n\geq 0\}$.
Further, let $G_m=(\{a\},\{P\},\omega,k)$ be a $\klPDOL$ system that generates the
language~$L_m$ and that is minimal with respect to the number of symbols. Then we
have $\omega=a^{5mk}$. 
From $\omega$, the word $a^{5mk+5m^2k}$ must be derived. Hence, the rule in $P$ 
is $a\to aa^{5m^2}$. With this rule, the number of $a$s is increased by $5m^2k$
in each step.
We have $\Symb_{\klPDOL}(L_m)=\Symb(G_m)=5mk+5m^2+3$.

The system $H_m=(\{a\},\{a\to aa^m\},a^{5mk},5mk)$ for $m\geq 1$ is a limited $\PDOL$ system
also generating~$L_m$.
We obtain $\Symb_{\lPDOL}(L_m)\leq 5mk+m+3$.
Hence, $\klPDOL\gg^{\Symb}\lPDOL$ for each $k\geq 1$.
\end{proof}

For the relations between the various types of $k$-limited and limited 
propagating L-systems, we give a results that covers them all.

\begin{theorem}\label{th-kPT0L-lPD0L}
The relation $\klPTOL\ggg \lPDOL$ is valid for $k\geq 1$.
\end{theorem}
\begin{proof}
Let $k\geq 1$ and $V=\Sets{a,b,c}$. For $m\in\N$, consider the language
\begin{multline*}
L_m=\Sets{c}\cup\{\:w \;|\; w=x_1x_2\cdots x_{(k+1)m},\ x_i\in\{a,bb\}, 1\leq i\leq (k+1)m,\\
|w|_a=j(k+1), 0\leq j\leq m \:\}.
\end{multline*}

As shown in the proof of Theorem~\ref{th-kPDT0L-kP0L}, the only possible rules 
for $a$ and $b$ are $a\ra a$, $a\ra bb$, and~$b\ra b$.
The axiom of a minimal $\klPOL$ system $G_m$ is $c$ and there is a rule~$c\ra a^{(k+1)m}$. 
If $G_m$ contains the rule $a\ra bb$, then words are derived that do not belong to $L_m$ 
(e.\,g., words with exactly $k$ subwords $bb$). Hence, the only rule for
$a$ is $a\ra a$ in any table. Thus, all words of the set $L_m\setminus\{c\}$ have to be 
derived directly from $c$. This yields more than $m$ rules and more than $m^2(k+1)$ symbols.

However, a $(k+1)$-limited $\PDOL$ system 
$H_m$ with the rules $a\ra bb$, $b\ra b$, and $c\ra a^{(k+1)m}$ 
also generates~$L_m$ but needs only three rules and $(k+1)m+10$ symbols.
This proves $\klPTOL\ggg^{\Prod} \lPDOL$ and $\klPTOL\gg^{\Symb} \lPDOL$.
\end{proof}

From this result, we obtain the relations 
\[k\ell X\ggg^{\Prod} \ell X \qmand k\ell X\gg^{\Symb} \ell X\]
for all classes \hbox{$X\in\{\POL,\PDTOL,\PTOL\}$}.

Summarizing, we found and proved relations between all classes of limited propagating L-systems
that were left open or that have not been considered in~\cite{Har09}. In some cases, we could
improve the results in~\cite{Har09} regarding propagating systems. The relations we stated here
are all tight.

Limited $\mathrm{T0L}$ systems that are not necessarily propagating have also been studied in~\cite{Har09}.
Except only a few cases, all relations between the various classes have been found and have been
proven to be tight. For some of the open questions, we can adopt results from the propagating
case to the non-propagating case.

\bibliographystyle{eptcs} 
\bibliography{dc_lsys}

\end{document}